\UseRawInputEncoding

\documentclass[reqno]{amsart}
\usepackage{mathrsfs}
\usepackage{graphicx}
\usepackage{amsmath}
\usepackage{amscd}
\usepackage{cite}
\usepackage{hyperref}
\usepackage{epsfig}
\usepackage{subfigure}
\usepackage{caption}
\usepackage{tikz}
  \usepackage{float}
  \usepackage{amssymb}
\usepackage{amsfonts}
\usepackage[all,cmtip]{xy}
\usepackage{appendix}
\usepackage{bm}

\usepackage{amsaddr}

\newtheorem{Example}{Example}[section]

\newtheorem{Theorem}{Theorem}[section]
\newtheorem{Theorem/Definition}{Theorem/Definition}[section]
\newtheorem{Proposition}{Proposition}[section]
\newtheorem{Lemma}{Lemma}[section]
\newtheorem{Corollary}{Corollary}[section]

\newcommand{\pd}{\partial}

\newcommand{\bC}{{\mathbb C}}

\newcommand{\bP}{{\mathbb P}}

\newcommand{\bZ}{{\mathbb Z}}

\newcommand{\cF}{{\mathcal F}}

\newcommand{\half}{\frac{1}{2}}

\newcommand{\wB}{{\widehat B}}

\newcommand{\hf}{{\hat f}}
\newcommand{\hF}{{\hat F}}

\newcommand{\be}{\begin{equation}}
\newcommand{\ee}{\end{equation}}
\newcommand{\bea}{\begin{eqnarray}}
\newcommand{\ben}{\begin{eqnarray*}}
\newcommand{\een}{\end{eqnarray*}}
\newcommand{\eea}{\end{eqnarray}}

\DeclareMathOperator{\Pf}{Pf}

\usepackage{color}

\definecolor{yellow}{rgb}{1,1,0}
\definecolor{orange}{rgb}{1,.7,0}
\definecolor{red}{rgb}{1,0,0}
\definecolor{green}{rgb}{0,1,1}
\definecolor{white}{rgb}{1,1,1}

\definecolor{A}{rgb}{.75,1,.75}

\theoremstyle{remark}
\newtheorem{Remark}{Remark}[section]

\begin{document}

\newtheorem{myDef}{Definition}
\newtheorem{thm}{Theorem}
\newtheorem{eqn}{equation}

\title[Connected Spin Double Hurwitz Numbers]
{Connected $(n,m)$-Point Functions of Diagonal $2$-BKP Tau-Functions,
and Spin Double Hurwitz Numbers}

\author{Zhiyuan Wang$^1$, $\quad$Chenglang Yang$^{1,2}$}
\address{\begin{flushleft}
1. School of Mathematical Sciences,
Peking University \\
2. Beijing International Center for Mathematical Research,
Peking University
\end{flushleft}
}
\email{zhiyuan19@math.pku.edu.cn, yangcl@pku.edu.cn}

\begin{abstract}

We derive an explicit formula for the connected $(n,m)$-point functions
associated to an arbitrary diagonal tau-function of the $2$-BKP hierarchy
using computation of neutral fermions and boson-fermion correspondence of type $B$,
and then apply this formula to the computation of connected spin double Hurwitz numbers.
This is the type $B$ analogue of [Z.Wang and C.Yang, arXiv:2210.08712 (2022)].

\end{abstract}

\maketitle

%\tableofcontents

\section{Introduction}

This is a sequel of our previous work \cite{wy2}.
In that paper,
we have derived an explicit formula for connected $(n,m)$-point functions of
diagonal tau-functions of the 2d Toda lattice hierarchy
using boson-fermion correspondence,
and then applied that formula to compute the connected double Hurwitz numbers
and stationary Gromov-Witten invariants of $\bP^1$ relative to two points.
Now in this work,
we deal with the case of diagonal tau-functions of $2$-BKP hierarchy,
and apply the result to the computation of connected spin double Hurwitz numbers.

Hurwitz numbers \cite{hur} count the numbers of branched covers between Riemann surfaces
with specified ramification types.
They provide a typical example of the interactions among various branches of mathematics,
including representation theory of symmetric groups,
intersection theory on moduli spaces,
integrable hierarchies,
topological recursions,
and combinatorics,
see e.g. \cite{elsv1, elsv2, op, op2, pa, dij, gv, gjv, Ok1}.
In particular,
Okounkov \cite{Ok1} showed that the generating series of (disconnected) double Hurwitz numbers
is a tau-function of the 2d Toda lattice hierarchy \cite{ta, ut},
and derived a fermionic representation of this tau-function
in terms of the charged free fermions  $\psi_r,\psi_r^*$ ($r\in \bZ +\half$).
See \cite{djm, jm, sa, sw} for introductions of Kyoto School's approach to integrable hierarchies
and the boson-fermion correspondence.
Okounkov's fermionic representation is the starting point of our computations in\cite{wy2}.

In \cite{eop},
Eskin-Okounkov-Pandharipande studied a new type of Hurwitz numbers
called the spin Hurwitz numbers,
by introducing a spin structure or a theta characteristic on the Riemann surface.
They are known to be related to the representations of the Sergeev group \cite{le2, gu}.
Similar to the cases of ordinary Hurwitz numbers,
the generating series of (disconnected) spin Hurwitz numbers are also known to be controlled by
some integrable hierarchies, see e.g. \cite{le, gkl, mmn, mmno, gkls}.
In this cases the corresponding  hierarchy is the BKP hierarchy \cite{djkm},
and the the boson-fermion correspondence of type $B$ is formulated in terms of the neutral fermions
$\phi_m$ ($m\in \bZ$).
See e.g. \cite{jm, hb, ost, tu, va, yo} for more about the BKP hierarchy.

For an even number $r>0$,
spin double Hurwitz numbers (indexed by two strict partitions) with completed $(r+1)$-cycles
count the branched covers between Riemann surfaces with spin structures
such that the ramification types over two branch points are described by
these two partitions,
and the types over all other branch points are completed $(r+1)$-cycles.
Roughly speaking,
these completed cycles are ramifications of order $r + 1$,
plus a certain linear combination of lower order ramifications that indicate degenerations of coverings,
see e.g. \cite{mmn, gkl} for details.
In \cite{gkl},
Giacchetto-Kramer-Lewa\'nski proved the polynomiality of such spin Hurwitz numbers
and proposed a spectral curve which conjecturally generates spin Hurwitz numbers via topological recursion.
In particular, they found a fermionic representation of the generating series
of disconnected spin double Hurwitz numbers with $(r+1)$-completed cycles,
and represented the connected spin double Hurwitz numbers with $(r+1)$-completed cycles
as summations over some commutation patterns in \cite[\S 6]{gkl},
inspired by \cite{jo, ssz}.

Now in this work,
we compute the connected spin double Hurwitz numbers with $(r+1)$-completed cycles
using a different approach.
We will consider more generally an arbitrary diagonal tau-function of the $2$-BKP hierarchy,
and derive a formula for the connected $(n,m)$-point functions
using computations of neutral fermions and the boson-fermion correspondence of type $B$.
Our method is inspired by Zhou's formula for connected $n$-point functions
associated to a KP tau-function \cite{zhou1} in terms of affine coordinates \cite{zhou3, by}.
See \cite{wy} for an analogue of Zhou's formula for the BKP hierarchy.
Our main result in this work is as follows.
Let
\begin{equation*}
\tau_f^B (\bm t^+,\bm t^-)
 = \langle \Gamma_+^B (\bm t^+) \exp(\hf) \Gamma_-^B (\bm t^-)
\rangle
\end{equation*}
be a diagonal tau-function of the $2$-BKP hierarchy,
where
\begin{equation*}
\hf =\sum_{m>0} (-1)^m f(m) :\phi_m\phi_{-m}:,
\end{equation*}
and
$f: \bZ_{>0} \to \bC$ is an arbitrary function
defined on the set of all positive integers.
Then the free energy $\log \tau_f^B (\bm t^+,\bm t^-)$ can be computed
by the following formula for connected $(n,m)$-point functions
(see Theorem \ref{thm-main-conn-nmpt-2} and Corollary \ref{cor-vanish-2}):
\be
\label{eq-intro-main}
\begin{split}
&\sum_{\substack{j_1,\cdots,j_n>0:\text{ odd}\\k_1,\cdots,k_m >0:\text{ odd} }}
\frac{\pd^{n+m} \log\tau_f^B(\bm t^+,\bm t^-)}
{\pd t_{j_1}^+\cdots\pd t_{j_n}^+ \pd t_{k_1}^-\cdots\pd t_{k_m}^-}
\Big|_{\bm t=0} \cdot
\prod_{a=1}^n z_a^{-j_a} \prod_{b=1}^m z_{n+b}^{k_b}\\
=& -2^{n+m-1} \cdot
\Big[ \sum_{\text{$(n+m)$-cycles $\sigma$} }
\prod_{i=1}^{n+m} \xi( z_{\sigma(i)}, - z_{\sigma(i+1)}) \Big]_{\text{odd}},
\qquad\quad
\forall m,n >0,
\end{split}
\ee
where the summation in the left-hand side is over all
positive odd integers $j_1,\cdots,j_n$, $k_1,\cdots,k_m$;
and $[\cdot]_{\text{odd}}$ means taking the terms of odd degrees in every $z_i$.
And for $\sigma(i) < \sigma(i+1)$,
$\xi$ is given by:
\begin{equation*}
\xi( z_{\sigma(i)}, - z_{\sigma(i+1)})\\
= \begin{cases}
-\frac{1}{4} -\half \sum_{k=1}^\infty
e^{ f(k)}z_{\sigma(i)}^{-k} z_{\sigma(i+1)}^k ,
& \text{ if $\sigma(i)\leq n<\sigma(i+1)$;}\\
-\frac{1}{4} -\half \sum_{k=1}^\infty
z_{\sigma(i)}^{-k} z_{\sigma(i+1)}^k ,
& \text{ otherwise,}
\end{cases}
\end{equation*}
and $\xi( z_{\sigma(i)}, - z_{\sigma(i+1)})= - \xi( - z_{\sigma(i+1)}, z_{\sigma(i)}) $
for $\sigma(i) > \sigma(i+1)$.
Here we use the convention $\sigma(n+m+1)=\sigma(1)$.

It is worth mentioning that
there has been another formula calculating connected correlators of such tau-functions in literature,
see Alexandrov-Shadrin \cite[Theorem 4.1]{AS}.
They treated one of the two families $\bm t^+$ and $\bm t^-$ as time variables
and the other as parameters,
and their formula is of a similar fashion as ours but looks more complicated.
It will be interesting to compare these two formulas.

Then we apply our formula to the
connected spin double Hurwitz numbers with $(r+1)$-completed cycles
$h_{g;\mu^+,\mu^-}^{\circ, r,\vartheta}$.
Here the superscript $\vartheta$,
which simply indicates the spin structure
(theta characteristic) in the definition,
 is a notation to distinguish spin Hurwitz numbers
from ordinary Hurwitz numbers.
Using the fermionic representation given in \cite{gkl},
we know that by taking the function $f$ to be $f^{r,\vartheta}(k) = \beta k^{r+1}/(r+1)$,
we obtain the following formula
(see \S \ref{sec-spinH}):
\begin{equation*}
h_{\mu^+,\mu^-}^{\circ, r,\vartheta}(\beta) = -
\frac{2^{n+m-1}}{Z_{\mu^+} Z_{\mu^-}}
\text{Coeff}_{\prod\limits_{a=1}^n z_a^{-\mu_a^+} \prod\limits_{b=1}^m z_{n+b}^{\mu_b^-}}
\bigg(
\sum_{\text{$(n+m)$-cycles} }
\prod_{i=1}^{n+m} \xi( z_{\sigma(i)}, - z_{\sigma(i+1)})
\bigg),
\end{equation*}
where $\text{Coeff}$ means taking the coefficient, and
\begin{equation*}
h_{\mu^+,\mu^-}^{\circ, r,\vartheta}(\beta) =
\sum_b 2^{g-1} \beta^b
\frac{h_{g;\mu^+,\mu^-}^{\circ, r,\vartheta}}
{l(\mu^+)! l(\mu^-)!},
\end{equation*}
Here $b$ is determined by the Riemann-Hurwitz formula $b= (2g-2 +l(\mu^+) +l(\mu^-)) /r$.

The rest of this paper is arranged as follows.
In \S \ref{sec-pre} we recall the preliminaries of boson-fermion correspondence of type $B$.
In \S \ref{sec-ferm-2n2m} we compute the fermionic $(2n,2m)$-point functions
associated to a diagonal tau-function,
and then in \S \ref{sec-disconn} we compute the disconnected bosonic $(n,m)$-point functions.
We recall the relation of connected bosonic $(n,m)$-point functions and the free energy
in \S \ref{sec-conn-boson},
and prove the formula \eqref{eq-intro-main} in \S \ref{sec-main}.
Finally in \S \ref{sec-spinH} we apply \eqref{eq-intro-main}
to the connected spin double Hurwitz numbers.

\section{Preliminaries}
\label{sec-pre}

In this section,
we recall some preliminaries of the neutral fermions and
boson-fermion correspondence of type $B$.
For details,
see \cite{yo, djkm, jm}.

First we recall the neutral fermions.
Let $\{\phi_m\}_{m\in \bZ}$ be a family of operators
satisfying the following anti-commutation relations:
\be
\label{eq-anticomm}
[\phi_m, \phi_n]_+ =(-1)^m \delta_{m+n,0},
\ee
where the bracket $[\cdot,\cdot]_+$ is defined by $[a,b]_+ = ab+ba$.
These operators are called the neutral fermions.
In particular, we have $\phi_0^2=\half$,
and $\phi_n^2=0$ for every $n\not=0$.

The fermionic Fock space $\cF_B$ of type $B$ is the $\bC$-vector space
of all formal (infinite) summations
\begin{equation*}
\sum
c_{k_1,\cdots,k_n}
\phi_{k_1} \phi_{k_2} \cdots \phi_{k_n} |0\rangle,
\qquad
c_{k_1,\cdots,k_n} \in\bC,
\end{equation*}
over $n\geq 0$ and $k_1>\cdots >k_n \geq 0$,
where $|0\rangle$ is a vector (called the fermionic vacuum vector) satisfying:
\be
\label{eq-B-anni}
\phi_i |0\rangle = 0,
\qquad \forall i <0.
\ee
The operators $\{\phi_n\}_{n\geq 0}$ are called the fermionic creators,
and $\{\phi_n\}_{n<0}$ are called the fermionic annihilators.
The Fock space $\cF_B$ can be decomposed as a direct sum of even and odd parts as follows:
\begin{equation*}
\cF_B = \cF_B^0 \oplus \cF_B^1,
\end{equation*}
where $\cF_B^0$ and $\cF_B^1$ are the subspaces with
even and odd numbers of the generators $\{\phi_i\}_{i\geq 0}$ respectively.

Recall that a partition of an integer $n$ is a sequence of integers $\mu=(\mu_1,\cdots,\mu_l)$
such that $\mu_1\geq \cdots \geq\mu_l>0$ and $|\mu|:=\mu_1+\cdots+\mu_n = n$.
A partition $\mu$ is called strict
if $\mu_1>\mu_2\cdots>\mu_l>0$.
The set of all strict partitions is denoted by $DP$,
and here we allow the empty partition $(\emptyset) \in DP$ of length zero.

The even part $\cF_B^0$ of the fermionic Fock space has a natural basis $\{|\mu\rangle\}_{\mu\in DP}$
indexed by all strict partitions.
Let $\mu\in DP $ be a strict partition $\mu= (\mu_1>\cdots >\mu_n > 0) $,
then then vector $|\mu\rangle$ is defined by:
\be
|\mu\rangle
= \begin{cases}
\phi_{\mu_1}\phi_{\mu_2}\cdots \phi_{\mu_n}|0\rangle, &\text{ for $n$ even;}\\
\sqrt{2}\cdot \phi_{\mu_1}\phi_{\mu_2}\cdots \phi_{\mu_n} \phi_0 |0\rangle, &\text{ for $n$ odd}.
\end{cases}
\ee
In particular, one has $|(\emptyset)\rangle =|0\rangle$ for the empty partition.

The dual Fock space $\cF_B^*$ is defined to be the vector space spanned by:
\begin{equation*}
\langle 0|
\phi_{k_n}  \cdots \phi_{k_2} \phi_{k_1},
\qquad
k_1 < k_2 < \cdots < k_n \leq 0,
\quad n\geq 0,
\end{equation*}
where $\langle 0|$ is a vector satisfying:
\be
\label{eq-B-anni-2}
\langle 0| \phi_i = 0,
\qquad \forall i >0.
\ee
Then there is a nondegenerate pairing $\cF_B^* \times \cF_B \to \bC$ determined by
the conditions \eqref{eq-anticomm}, \eqref{eq-B-anni}, \eqref{eq-B-anni-2},
and the requirements
$\langle 0 | 0 \rangle =1$ and $\langle 0|\phi_0 | 0\rangle =0$.
Then for an arbitrary sequence $k_1>k_2>\cdots >k_n\geq 0$,
one has:
\be
\label{eq-basic-VEV}
\langle 0 | \phi_{-k_n}\cdots \phi_{-k_1}
\phi_{k_1}\cdots \phi_{k_n} |0\rangle = \begin{cases}
(-1)^{k_1+\cdots +k_n}, &\text{ if $k_n\not=0$;}\\
\half\cdot (-1)^{k_1+\cdots +k_{n-1}},&\text{ if $k_n=0$.}
\end{cases}
\ee
In general,
the vacuum expectation value of a product of neutral fermions
can be computed using Wick's Theorem:
\begin{equation*}
\langle 0| \phi_{i_1}\phi_{i_2}\cdots \phi_{i_{2n}}|0\rangle
=\sum_{\substack{(p_1,q_1,\cdots,p_n,q_n)\\p_k<q_k, \quad p_1<\cdots<p_n}}
\text{sgn}(p,q)\cdot \prod_{j=1}^n
\langle 0| \phi_{i_{p_j}}\phi_{i_{q_j}} |0\rangle,
\end{equation*}
where $(p_1,q_1,\cdots,p_n,q_n)$ is a permutation of $(1,2,\cdots,2n)$,
and $\text{sgn}(p,q)$ denotes its sign ($\text{sgn}=1$ for an even permutation, and $\text{sgn}=-1$ for an odd one).
In what follows,
we will denote by
\begin{equation*}
\langle A \rangle = \langle 0 |A|0\rangle
\end{equation*}
the vacuum expectation value of an operator $A$ on the fermionic Fock space.

The normal-ordered product $:\phi_i\phi_j:$ of two neutral fermions
is defined by:
\be
\label{eq-def-normal-order}
:\phi_i\phi_j: = \phi_i\phi_j -\langle \phi_i\phi_j \rangle.
\ee
In particular,
one has $:\phi_0^2: = \phi_0^2 -\langle 0|\phi_0^2|0\rangle=0$.
The relation \eqref{eq-anticomm} is equivalent to the following
operator product expansion:
\be
\label{eq-normal-fermfield}
\phi(w)\phi(z)
= :\phi(w)\phi(z): +i_{w,z}\frac{w-z}{2(w+z)},
\ee
where $\phi(z)$ is the fermionic field:
\be
\label{eq-ferm-field}
\phi(z) = \sum_{i\in \bZ} \phi_i z^i,
\ee
and $i_{w,z}$ means formally expanding on $\{|w|>|z|\}$, i.e.,
\begin{equation*}
i_{w,z}\frac{w-z}{2(w+z)}= \half +
\sum_{j=1}^\infty (-1)^{j} w^{-j} z^j.
\end{equation*}

Now we recall the construction of bosonic operators of type $B$.
Let $n\in 2\bZ+1$ be an odd integer,
and define the Hamiltonian $H_n$ by:
\be
\label{eq-def-boson}
H_n = \half \sum_{i\in \bZ} (-1)^{i+1} \phi_i\phi_{-i-n}.
\ee
One can check that they satisfy the following commutation relation:
\be
[H_n,H_m] = H_nH_m-H_mH_n= \frac{n}{2}\cdot \delta_{m+n,0},
\qquad
\forall n,m \text{ odd}.
\ee
The operators $H_n$ are called the bosons of type $B$.
And one has:
\begin{equation*}
H_n |0\rangle = 0,
\qquad\qquad
\forall n> 0.
\end{equation*}

Denote by $H(z)$ the generating series of these bosons:
\be
H(z)= \sum_{n\in \bZ: \text{ odd}} H_n z^{-n}
\ee
Notice that one can also define $H_{2k}$ using \eqref{eq-def-boson},
and the anti-commutation relation \eqref{eq-anticomm} implies
$H_{2k}=0$ for every $k\not= 0$.
In this sense,
one easily finds that:
\be
H(z) = -\half :\phi(-z)\phi(z):.
\ee

Now we recall the boson-fermion correspondence of type $B$.
Let $\bm t=(t_1,t_3,t_5,t_7,\cdots)$ be a family of formal variables,
and define:
\be
\label{eq-def-Gamma+-}
\begin{split}
&\Gamma_+^B(\bm t) = \exp\Big( \sum_{n>0: \text{ odd}} t_n H_n \Big)
= \exp\Big(\half \sum_{n>0: \text{ odd}} t_n \sum_{i\in \bZ} (-1)^{i+1} \phi_i\phi_{-i-n}\Big),\\
&\Gamma_-^B(\bm t) = \exp\Big( \sum_{n>0: \text{ odd}} t_n H_{-n} \Big)
=\exp\Big( \half \sum_{n>0: \text{ odd}} t_n \sum_{i\in \bZ} (-1)^{i+1} \phi_i\phi_{-i+n}\Big),
\end{split}
\ee
then the boson-fermion correspondence of type $B$ is:
\begin{Theorem}
[\cite{djkm}]
The following map is a linear isomorphism:
\begin{equation*}
\sigma_B :\cF_B \to \bC[\![w;t_1,t_2,\cdots]\!]/\sim,
\qquad
|U\rangle \mapsto \sum_{i=0}^1 \omega^i\cdot \langle i |\Gamma_+^B(\bm t)|U\rangle,
\end{equation*}
where $\omega^2\sim 1$,
and $\langle 1| =\sqrt{2}\langle 0|\phi_0 \in (\cF_B^1)^*$.
Under this isomorphism,
one has:
\be
\label{eq-bfcor-boson}
\sigma_B (H_n |U\rangle) = \frac{\pd}{\pd t_n} \sigma_B(|U\rangle),
\qquad
\sigma_B (H_{-n}|U\rangle)=\frac{n}{2} t_n \cdot \sigma_B(|U\rangle),
\ee
for every odd $n>0$.
Moreover,
\be
\label{eq-bfcor-ferm}
\sigma_B (\phi(z)|U\rangle) =
\frac{1}{\sqrt{2}} \omega\cdot e^{\xi(\bm t,z)} e^{-\xi (\tilde\pd ,z^{-1})} \sigma_B(|U\rangle),
\ee
where
$\xi(\bm t,z) = \sum_{n>0\text{ odd}} t_{n}z^{n}$
and $\tilde\pd = (2\pd_{t_{ 1}},\frac{2}{3}\pd_{t_{ 3}},\frac{2}{5}\pd_{t_{ 5}},\cdots)$.
\end{Theorem}

\section{Diagonal Tau-Functions and Fermionic $(2n,2m)$-Point Functions}
\label{sec-ferm-2n2m}

In this section,
we first recall the construction of diagonal tau-functions $\tau_f^B$
associated to a function $f:\bZ_{>0} \to \bC$,
and then compute the fermionic $(2n,2m)$-point functions
using Wick's Theorem for neutral fermions..

Let $\Gamma_\pm^B$ be the operators defined by \eqref{eq-def-Gamma+-},
and let $\bm t^\pm = (t_1^\pm,t_3^\pm,t_5^\pm,\cdots)$ be two families of formal variables.
A diagonal tau-function is a formal power series in $\bm t^\pm$ of the following form:
\be
\label{eq-def-tauf}
\tau_f^B (\bm t^+,\bm t^-)
 = \langle \Gamma_+^B (\bm t^+) \exp(\hf) \Gamma_-^B (\bm t^-)
\rangle,
\ee
where $\hf$ is the following operator on the fermionic Fock space:
\be
\hf =\sum_{m>0} (-1)^m f(m) :\phi_m\phi_{-m}:
= \sum_{m>0} (-1)^m f(m) \phi_m\phi_{-m},
\ee
and
\be
f: \bZ_{>0} \to \bC
\ee
is an arbitrary function on the set of positive integers.
One can check that the operator $e^\hf$ satisfies the Hirota bilinear equation:
\begin{equation*}
[e^\hf \otimes e^\hf , \sum_{m\in \bZ} (-1)^m \phi_m\otimes \phi_{-m}] = 0,
\end{equation*}
thus $\tau_f^B$ is a tau-function of the $2$-BKP hierarchy,
see e.g. \cite[Appendix]{ost}.

Let $z_1,z_2,\cdots,z_{2n+2m}$
be a family of formal variables.
The fermionic $(2n,2m)$-point function
associated to a diagonal tau-function $\tau_f^B$
is defined to be the following vacuum expectation value:
\be
\label{eq-de-ferm-nmpt}
\langle \phi(z_1)\phi(z_2)\cdots \phi(z_{2n})
e^\hf
\phi(z_{2n+1})\phi(z_{2n+2})\cdots \phi(z_{2n+2m}) \rangle.
\ee
Notice that by \eqref{eq-B-anni} we have:
\be
\exp(\hf) |0\rangle = |0\rangle,
\qquad\qquad
\langle 0| \exp(-\hf) = \langle 0|,
\ee
thus \eqref{eq-de-ferm-nmpt} can be rewritten as:
\be
\label{eq-de-ferm-nmpt-2}
\langle \phi_f(z_1)\phi_f(z_2)\cdots \phi_f(z_{2n})
\phi(z_{2n+1})\phi(z_{2n+2})\cdots \phi(z_{2n+2m}) \rangle,
\ee
where $\phi_f(z) $ is defined by:
\be
\phi_f(z) = e^{-\hf} \phi(z) e^{\hf}.
\ee

\begin{Lemma}
\label{lem-conj-hf-phi}
The field $\phi_f(z) = e^{-\hf} \phi(z) e^{\hf}$ is given by:
\be
\phi_f (z) =\sum_{k<0}e^{ f(-k)} \phi_k z^k
+\phi_0 + \sum_{k>0}e^{- f(k)} \phi_k z^k.
\ee
\end{Lemma}
\begin{proof}
Notice that by the anti-commutation relation \eqref{eq-anticomm} we have:
\begin{equation*}
\begin{split}
[\hf,\phi_k] = & \sum_{j>0} (-1)^j f(j)
[\phi_j \phi_{-j},\phi_k]\\
=& \sum_{j>0} (-1)^j f(j)
\big(\phi_j [ \phi_{-j},\phi_k]_+
- [\phi_j , \phi_k]_+ \phi_{-j}\big)\\
=& \sum_{j>0} f(j)
\big( \delta_{j,k} \phi_j - \delta_{j+k,0}\phi_{-j} \big)
\end{split},
\end{equation*}
thus:
\begin{equation*}
[\hf,\phi_k]
=\begin{cases}
f(k) \phi_k, & \text{ if $k>0$;}\\
- f(-k) \phi_k, & \text{ if $k<0$;}\\
0, & \text{ if $k=0$.}\\
\end{cases}
\end{equation*}
Then by the Baker-Campbell-Hausdorff formula we have:
\begin{equation*}
\begin{split}
e^{-\hf} \phi_k e^\hf
=& \phi_k - [\hf,\phi_k] + \frac{1}{2!}[\hf,[\hf,\phi_k]]
+ \frac{1}{3!} [\hf,[\hf,[\hf,\phi_k]]] -\cdots\\
=&\begin{cases}
e^{- f(k)} \phi_k, & \text{ if $k>0$;}\\
e^{ f(-k)} \phi_k, & \text{ if $k<0$;}\\
\phi_0, & \text{ if $k=0$.}
\end{cases}
\end{split}
\end{equation*}
Thus the conclusion holds.
\end{proof}

Then we have the following:
\begin{Theorem}
The fermionic $(2n,2m)$-point function \eqref{eq-de-ferm-nmpt} equals to
the Pfaffian:
\be
\label{eq-thm-ferm2n2m}
\langle \phi(z_1)\cdots \phi(z_{2n})
e^\hf
\phi(z_{2n+1})\cdots \phi(z_{2n+2m}) \rangle
= \Pf(\wB_{i,j}),
\ee
where $(\wB_{i,j})$ is an  anti-symmetric matrix of size $(2n+2m)\times (2n+2m)$,
whose upper-triangular part is given by:
\be
\label{eq-thm-ferm2n2m-defB}
\wB_{i,j} = \begin{cases}
\half +
\sum\limits_{k=1}^\infty (-1)^k e^{ f(k)} z_i^{-k}z_j^k,
& \text{ if $i\leq 2n<j$;}\\
i_{z_i,z_j} \frac{z_i-z_j}{2(z_i+z_j)},
& \text{ if $i<j\leq 2n$ or $2n+1\leq i<j$.}
\end{cases}
\ee
\end{Theorem}
\begin{proof}
Since $\phi_f(z)$ and $\phi (z)$ are both linear in $\{\phi_i\}$,
one can apply Wick's Theorem to the fermionic $(2n,2m)$-point function \eqref{eq-de-ferm-nmpt-2}
and get:
\be
\label{eq-wicksum-ferm2n2m}
\begin{split}
& \langle \phi_f(z_1)\phi_f(z_2)\cdots \phi_f(z_{2n})
\phi(z_{2n+1})\phi(z_{2n+2})\cdots \phi(z_{2n+2m}) \rangle \\
=&
\sum_{\substack{(p_1,q_1,\cdots,p_{n+m},q_{n+m})\\p_k<q_k, \quad p_1<\cdots<p_{n+m}}}
\text{sgn}(p,q)\cdot \prod_{j=1}^{n+m}
\langle  \tilde\phi(z_{p_j}) \tilde\phi(z_{q_j}) \rangle,
\end{split}
\ee
where $(p_1,q_1,\cdots,p_{n+m},q_{n+m})$ is a permutation of $(1,2,\cdots,2n+2m)$,
and
\be
\tilde\phi(z_i)=\begin{cases}
\phi_f (z_i), & \text{ if $1\leq i\leq 2n$;}\\
\phi (z_i), & \text{ if $2n+1\leq i \leq 2n+2m$.}
\end{cases}
\ee
This is equivalent to say that \eqref{eq-de-ferm-nmpt-2} equals to a Pfaffian:
\begin{equation*}
\langle \phi_f(z_1)\cdots \phi_f(z_{2n})
\phi(z_{2n+1})\cdots \phi(z_{2n+2m}) \rangle = \Pf
(\wB_{i,j})_{1\leq i,j\leq 2n+2m},
\end{equation*}
where the $(2n+2m)\times (2n+2m)$ matrix $(\wB_{i,j})$ is given by:
\be
\wB_{i,j} = \begin{cases}
\langle \tilde\phi(z_i)\tilde\phi(z_j) \rangle,
& \text{ if $1\leq i<j\leq 2n+2m$;}\\
-\langle \tilde\phi(z_j)\tilde\phi(z_i) \rangle,
& \text{ if $1\leq j<i\leq 2n+2m$;}\\
0, & \text{ if $1\leq i=j \leq 2n+2m$.}
\end{cases}
\ee
Notice that for $2n+1 \leq i<j\leq 2n+2m$, we have:
\begin{equation*}
\langle \tilde\phi(z_i)\tilde\phi(z_j) \rangle
= \langle \phi(z_i)\phi(z_j) \rangle
= \half +\sum_{k=1}^\infty (-1)^k z_i^{-k} z_j^k
= i_{z_i,z_j} \frac{z_i-z_j}{2(z_i+z_j)};
\end{equation*}
and for $1 \leq i \leq 2n < j\leq 2n+2m$, by Lemma \ref{lem-conj-hf-phi} we have:
\begin{equation*}
\begin{split}
\langle \tilde\phi(z_i)\tilde\phi(z_j) \rangle
=& \big\langle \big(\sum_{k<0}e^{ f(-k)} \phi_k z_i^k
+\phi_0 + \sum_{k>0}e^{- f(k)} \phi_k z_i^k \big)\phi(z_j) \big\rangle\\
=& \half +
\sum_{k=1}^\infty (-1)^k e^{ f(k)} z_i^{-k}z_j^k;
\end{split}
\end{equation*}
and for $1 \leq i < j\leq 2n$, by Lemma \ref{lem-conj-hf-phi} we have:
\begin{equation*}
\begin{split}
\langle \tilde\phi(z_i)\tilde\phi(z_j) \rangle
=& \big\langle \big(\sum_{k<0}e^{ f(-k)} \phi_k z_i^k
+\phi_0 + \sum_{k>0}e^{- f(k)} \phi_k z_i^k \big) \\
& \quad \big(\sum_{k<0}e^{ f(-k)} \phi_k z_j^k
+\phi_0 + \sum_{k>0}e^{- f(k)} \phi_k z_j^k \big)
\big\rangle\\
=& \half +
\sum_{k=1}^\infty (-1)^k
 z_i^{-k}z_j^k .
\end{split}
\end{equation*}
Then the conclusion is clear.
\end{proof}

\section{Computation of Disconnected Bosonic $(n,m)$-Point Functions}
\label{sec-disconn}

In this section,
we use the results of the previous section and boson-fermion correspondence of type $B$
to compute the following (disconnected) bosonic $(n,m)$-point functions
associated to $\tau_f^B$:
\be
\label{eq-de-ferm-nmpt}
\langle H(z_1)H(z_2)\cdots H(z_{n})
e^\hf
H(z_{n+1})H(z_{n+2})\cdots H(z_{n+m}) \rangle.
\ee

Let $f:\bZ_{>0} \to \bC$ be a function,
and denote by $A_f$ the following formal series:
\be
\label{eq-def-Af}
A_f (w,z) = -\frac{1}{4} -\half \sum_{k=1}^\infty
(-1)^k e^{ f(k)}w^{-k}z^k.
\ee
In particular,
one has:
\be
A_0 (w,z) = -\half \cdot i_{w,z}\frac{w-z}{2(w+z)}.
\ee
First we prove the following:
\begin{Proposition}
We have:
\be
\label{eq-normalferm-2n2m}
\begin{split}
&\langle :\phi(z_1)\phi(z_2): :\phi(z_3)\phi(z_4): \cdots :\phi(z_{2n-1})\phi(z_{2n}):
e^\hf \\
& :\phi(z_{2n+1})\phi(z_{2n+2}) : \cdots :\phi(z_{2n+2m-1}) \phi(z_{2n+2m}): \rangle
= \Pf (B_{i,j}),
\end{split}
\ee
where $(B_{i,j})$ is an  anti-symmetric matrix of size $(2n+2m)\times (2n+2m)$,
whose upper-triangular part is given by:
\be
B_{i,j} = \begin{cases}
-2A_f(z_i,z_j),
& \text{ if $i\leq 2n<j$;}\\
0, & \text{ if $i=2s-1,j=2s$ for some $s$;}\\
-2 A_0(z_i,z_j),
& \text{ other cases where $i<j$.}
\end{cases}
\ee
\end{Proposition}
\begin{proof}
By \eqref{eq-normal-fermfield} we know that
the left-hand side of \eqref{eq-normalferm-2n2m} equals to:
\be
\label{eq-normalferm-2n2m-2}
\begin{split}
&\big\langle \big( \phi(z_1)\phi(z_2) - \frac{z_1-z_2}{2(z_1+z_2)}\big)
 \cdots
\big( \phi(z_{2n-1})\phi(z_{2n}) -  \frac{z_{2n-1}-z_{2n}}{2(z_{2n-1}+z_{2n})} \big)
 \\
& e^\hf
\big( \phi(z_{2n+1})\phi(z_{2n+2}) -\frac{z_{2n+1}-z_{2n+2}}{2(z_{2n+1}+z_{2n+2})} \big) \cdots\\
& \big(\phi(z_{2n+2m-1}) \phi(z_{2n+2m}) - \frac{z_{2n+2m-1}-z_{2n+2m}}{2(z_{2n+2m-1}+z_{2n+2m})}
\big) \big\rangle\\
=&
\big\langle \big( \phi_f(z_1)\phi_f(z_2) - \frac{z_1-z_2}{2(z_1+z_2)}\big)
 \cdots
\big( \phi_f(z_{2n-1})\phi_f(z_{2n}) -  \frac{z_{2n-1}-z_{2n}}{2(z_{2n-1}+z_{2n})} \big)
 \\
& \big( \phi(z_{2n+1})\phi(z_{2n+2}) -\frac{z_{2n+1}-z_{2n+2}}{2(z_{2n+1}+z_{2n+2})} \big) \cdots\\
& \big(\phi(z_{2n+2m-1}) \phi(z_{2n+2m}) - \frac{z_{2n+2m-1}-z_{2n+2m}}{2(z_{2n+2m-1}+z_{2n+2m})}
\big) \big\rangle \\
=& \sum_{K_1\sqcup L_1 =[n] , K_2\sqcup L_2 =[m]}
\Big( \prod_{l \in L_1}f_l \Big) \Big( \prod_{l \in L_2}f_l' \Big)
\langle \phi_{K_1} \phi_{K_2}' \rangle,
\end{split}
\ee
where $[n]$ denotes the set $\{1,2,\cdots,n\}$, and
\begin{equation*}
f_l = - \frac{z_{2l-1}-z_{2l}}{2(z_{2l-1}+z_{2l})},
\qquad
f_l' = - \frac{z_{2n+2l-1}-z_{2n+2l}}{2(z_{2n+2l-1}+z_{2n+2l})},
\end{equation*}
and for a set $K=\{k_1,k_2,\cdots,k_s\}$ with $k_1<k_2<\cdots <k_s$,
\begin{equation*}
\begin{split}
&\phi_K= \phi_f(z_{2k_1-1})\phi_f(z_{2k_1}) \phi_f(z_{2k_2-1})\phi_f(z_{2k_2})
\cdots \phi_f(z_{2k_s-1})\phi_f(z_{2k_s}),\\
&\phi_K'= \phi(z_{2n+2k_1-1})\phi(z_{2n+2k_1})
\cdots \phi(z_{2n+2k_s-1})\phi(z_{2n+2k_s}).
\end{split}
\end{equation*}

On the other hand, notice that:
\begin{equation*}
B_{i,j} = \begin{cases}
\wB_{i,j} + f_s,
& \text{ if $i=2s-1,j=2s$ for some $1\leq s\leq n$;}\\
\wB_{i,j} + f_s',
& \text{ if $i=2s-1,j=2s$ for some $n+1\leq s \leq n+m$;}\\
\wB_{i,j}, & \text{ otherwise.}
\end{cases}
\end{equation*}
Therefore the right-hand side of \eqref{eq-normalferm-2n2m} equals to:
\begin{equation*}
\begin{split}
\Pf (B_{i,j}) =&
\sum_{\substack{(p_1,q_1,\cdots,p_{n+m},q_{n+m})\\p_k<q_k, \quad p_1<\cdots<p_{n+m}}}
\text{sgn}(p,q)\cdot \prod_{k=1}^{n+m}
B_{p_k,q_k}\\
=&
\sum_{(p,q)} \text{sgn}(p,q)
\prod_{k\in A_{(p,q)}} (\wB_{p_k,q_k}+f_{\frac{q_k}{2}})
\prod_{k\in B_{(p,q)}} (\wB_{p_k,q_k}+f_{\frac{q_k}{2}}')
\prod_{k\in C_{(p,q)}} \wB_{p_k,q_k},
\end{split}
\end{equation*}
where the summation is over all possible  permutations
$(p_1,q_1,\cdots,p_{n+m},q_{n+m})$ of $(1,2,\cdots,2n+2m)$,
and $\text{sgn}(p,q) = \pm 1$ denotes the sign of this permutation,
and $A_{(p,q)},B_{(p,q)},C_{(p,q)}$ are the following partition of the set $\{1,2, \cdots,n+m\}$:
\begin{equation*}
\begin{split}
&A_{(p,q)} = \{k \big| \text{$p_k = 2s-1$, $q_k=2s$, for some $1\leq s \leq n$ }\}; \\
&B_{(p,q)} = \{k \big| \text{$p_k = 2s-1$, $q_k=2s$, for some $n+1\leq s \leq n+m$ }\}; \\
&C_{(p,q)} = \{1,2,\cdots,n+m\} \backslash (A_{(p,q)} \cup B_{(p,q)}).
\end{split}
\end{equation*}
Then we expand the product in the right-hand side and obtain:
\be
\label{eq-pf-new}
\begin{split}
\Pf (B_{i,j}) =&
\sum_{L_1\subset [n], L_2\subset [m]}
\prod_{l\in L_1} f_{l} \prod_{l\in L_2} f'_{l}
\sum_{(p',q')} \text{sgn}(p',q') \prod_k \wB_{p_k',q_k'}.
\end{split}
\ee
Here $(p',q')=(p_1',q_1',p_2',q_2',\cdots)$ in the right-hand side runs over all permutations of
$\{1,2,\cdots,2n+2m\} \backslash
\big( \widetilde L_1 \cup  \widetilde L_2 \big)$
such that $p_k'<q_k'$ and $p_1'<p_2'<\cdots$,
where
\begin{equation*}
\begin{split}
&\widetilde L_1 = \{ 2l-1 \big| l\in L_1 \} \cup \{ 2l \big| l\in L_1\},\\
&\widetilde L_2 = \{ 2l+2n-1 \big| l\in L_2 \} \cup \{ 2l+2n \big| l\in L_2\}.
\end{split}
\end{equation*}
Here we have also used the fact that deleting two adjacent indices $(2s-1,2s)$
in a permutation preserves the sign.
And again by Wick's Theorem we see that the right-hand side of \eqref{eq-pf-new}
equals to \eqref{eq-normalferm-2n2m-2}.
This completes the proof.
\end{proof}

Recall that the generating series $H(z)$ of bosons are given by
(see \S \ref{sec-pre}):
\begin{equation*}
H(z) = -\half :\phi(-z) \phi(z):,
\end{equation*}
thus by taking $z_{2s-1}\to -z_{2s}$ for every $1\leq s \leq n+m$
in the above proposition,
we obtain our main result in this section:
\begin{Theorem}
\label{thm-disconn-b-nmpt}
The bosonic $(n,m)$-point functions are given by:
\be
\langle H(z_1)\cdots H(z_{n})
e^\hf
H(z_{n+1})\cdots H(z_{n+m}) \rangle
=\Pf (C_{i,j}),
\ee
where $C_{i,j}$ is an  anti-symmetric matrix of size $(2n+2m)\times (2n+2m)$,
whose upper-triangular part is given by:
\be
C_{i,j} = \begin{cases}
A_f \big(
(-1)^i z_{\lceil \frac{i}{2}\rceil }, (-1)^j z_{\lceil \frac{j}{2} \rceil}
\big),
& \text{ if $i\leq 2n<j$;}\\
0, & \text{ if $i=2s-1,j=2s$ for some $s$;}\\
A_0 \big(
(-1)^i z_{\lceil \frac{i}{2}\rceil }, (-1)^j z_{\lceil \frac{j}{2} \rceil}
\big),
& \text{ other cases where $i<j$.}
\end{cases}
\ee
\end{Theorem}

\section{Connected Bosonic $(n,m)$-Point Functions and Free Energy}
\label{sec-conn-boson}

In this section,
we first recall the notion of bosonic $(n,m)$-point functions
associated to the tau-function $\tau_f^B$,
and then discuss the relation between them and the free energy $\log \tau_f^B$.
The computations of these connected bosonic $(n,m)$-point functions
will be carried out in next section.

Let $f:\bZ_{>0} \to \bC$ be an arbitrary function,
and let $\tau_f^B$ be the diagonal tau-function given by \eqref{eq-def-tauf}.
In what follows we will denote by
\be
\langle H(z_{[n+m]}) \rangle_{f;n,m}
= \langle H(z_1)\cdots H(z_{n})
e^\hf
H(z_{n+1})\cdots H(z_{n+m}) \rangle
\ee
the (disconnected) bosonic $(n,m)$-point function associated to $\tau_f^B$,
where $[n+m]$ denotes the set $\{1,2,\cdots,n+m\}$,
and we use the notation
\begin{equation*}
H(z_I) = H(z_{i_1})H(z_{i_2})\cdots H(z_{i_k})
\end{equation*}
for $I=\{i_1,i_2,\cdots,i_k\}$ with $i_1< i_2< \cdots<i_k$.
The connected $(n,m)$-point function
$\langle H(z_1) \cdots H(z_n) H(z_{n+1}) \cdots H(z_{n+m}) \rangle_{f;n,m}^c$
is defined by the following M\"obius inversion formulas:
\begin{equation*}
\begin{split}
&\langle H(z_{[n+m]})  \rangle_{f;m,n}
= \sum_{I_1\sqcup \cdots \sqcup I_k =[n+m]} \frac{1}{k!}
\langle H(z_{I_1}) \rangle_{f;n_1,m_1}^c \cdots \langle H(z_{I_k}) \rangle_{f;n_k,m_k}^c,\\
&\langle H(z_{[n+m]})  \rangle_{f;m,n}^c
=\sum_{I_1\sqcup \cdots \sqcup I_k =[n+m]} \frac{(-1)^{k-1}}{k}
\langle H(z_{I_1}) \rangle_{f;n_1,m_1} \cdots \langle H(z_{I_k}) \rangle_{f;n_k,m_k},
\end{split}
\end{equation*}
where $I_1,\cdots,I_k \subset [n+m]$ are nonempty subsets,
and we denote:
\be
\label{eq-notation-njmj}
n_j = \big|I_j\cap [n]\big|,
\qquad
m_j = \big|I_j \backslash [n]\big|,
\qquad
1\leq j\leq k.
\ee

Now denote by
\be
F_f^B(\bm t^+,\bm t^-) = \log \tau_f^B (\bm t^+,\bm t^-)
\ee
the free energy associated to $\tau_f^B$.
In the rest of this section,
we show that:
\begin{Proposition}
\label{prop-relation-fe}
For every pair $(n,m)$ with $n+m\geq 1$,
we have:
\be
\label{eq-reln-npt&fe}
\begin{split}
&\langle H(z_{[n+m]}) \rangle_{f;n,m}^c
= (\delta_{n,2}\delta_{m,0} + \delta_{n,0}\delta_{m,2} )\cdot
i_{z_1,z_2} \frac{z_1z_2 (z_1^2+z_2^2)}{2(z_1^2-z_2^2)^2}\\
&+ \sum_{\substack{j_1,\cdots,j_n>0:\text{ odd}\\k_1,\cdots,k_m >0:\text{ odd} }}
\frac{\pd^{n+m} F_f^B(\bm t^+,\bm t^-)}
{\pd t_{j_1}^+\cdots\pd t_{j_n}^+ \pd t_{k_1}^-\cdots\pd t_{k_m}^-}
\Big|_{\bm t=0} \cdot
\prod_{a=1}^n z_a^{-j_a} \prod_{b=1}^m z_{n+b}^{k_b},
\end{split}
\ee
where $\bm t=(\bm t^+,\bm t^-)$,
and
\be
i_{z_1,z_2} \frac{z_1z_2 (z_1^2+z_2^2)}{2(z_1^2-z_2^2)^2}
= \sum_{i>0: \text{ odd}} \frac{i}{2}z_1^{-i} z_2^i.
\ee
\end{Proposition}
\begin{proof}
Define the functions $G_{f;n,m}$ to be:
\begin{equation*}
\begin{split}
&G_{f;n,m}(z_1,\cdots,z_{n+m}) \\
=& \frac{1}{\tau_f^B (\bm t^+,\bm t^-)}
\langle \Gamma_+^B(\bm t^+)
H(z_1)\cdots H(z_n) e^\hf H(z_{n+1})\cdots H(z_{n+m})
\Gamma_-^B(\bm t^-) \rangle,
\end{split}
\end{equation*}
and define the connected version $G_{f;n,m}^c$ by the M\"obius inversion:
\begin{equation*}
G_{f;n,m}^c (z_1,\cdots,z_{n+m})
=\sum_{I_1\sqcup \cdots \sqcup I_k =[n+m]} \frac{(-1)^{k-1}}{k}
G_{f;n_1,m_1}(z_{I_1})\cdots G_{f;n_k,m_k}(z_{I_k}),
\end{equation*}
where we use the notation \eqref{eq-notation-njmj}.
Then we have:
\begin{equation*}
\begin{split}
&\langle H(z_{[n+m]}) \rangle_{f;m,n} =
G_{f;n,m}(z_1,\cdots,z_{n+m})|_{\bm t =0},\\
&\langle H(z_{[n+m]}) \rangle_{f;m,n}^c =
G_{f;n,m}^c(z_1,\cdots,z_{n+m})|_{\bm t =0}.
\end{split}
\end{equation*}

For every pair $(n,m)$ with $n+m\geq 3$,
we have:
\begin{equation*}
\begin{split}
G_{f;n,m}^c(z_1,\cdots,z_{n+m})=
\sum_{\substack{j_1,\cdots,j_n>0:\text{ odd}\\k_1,\cdots,k_m >0:\text{ odd} }}
\frac{\pd^{n+m} F_f^B(\bm t^+,\bm t^-)}
{\pd t_{j_1}^+\cdots\pd t_{j_n}^+ \pd t_{k_1}^-\cdots\pd t_{k_m}^-}
\prod_{a=1}^n z_a^{-j_a} \prod_{b=1}^m z_{n+b}^{k_b}.
\end{split}
\end{equation*}
This is actually \cite[Prop. 5.1]{zhou1},
see also \cite[Prop. 4.1]{wy2},
and here we will not repeat the proof.
Then the relation \eqref{eq-reln-npt&fe} for $n+m\geq 3$
follows by taking $\bm t=0$.

In what follows,
we check the relation \eqref{eq-reln-npt&fe} for
$(n,m)$ with $n+m\leq 2$ directly.
First consider the case $(n,m)=(1,0)$.
By the boson-fermion correspondence \eqref{eq-bfcor-boson} we know that
for every odd integer $k>0$,
\begin{equation*}
\begin{split}
&\langle \Gamma_+(\bm t^+) H_k e^\hf \Gamma_-(\bm t^-) \rangle
= \frac{\pd}{\pd t_k^+} \tau_f^B,\qquad
\langle \Gamma_+(\bm t^+) H_{-k} e^\hf \Gamma_-(\bm t^-) \rangle
= \frac{n}{2}t_k^+ \tau_f^B;\\
&\langle \Gamma_+(\bm t^+) e^\hf H_k \Gamma_-(\bm t^-) \rangle
= \frac{n}{2}t_k^- \tau_f^B,\qquad
\langle \Gamma_+(\bm t^+) e^\hf H_{-k} \Gamma_-(\bm t^-) \rangle
= \frac{\pd}{\pd t_k^-} \tau_f^B.
\end{split}
\end{equation*}
Thus we have:
\begin{equation*}
\begin{split}
G_{f;1,0}(z)
=& \frac{1}{\tau_f^B} \cdot
\sum_{k>0:\text{ odd}} \big(
z^{-k}\frac{\pd}{\pd t_k^+} + \frac{k}{2} t_k^+\cdot z^k
\big) \tau_f^B \\
= &
\sum_{k>0:\text{ odd}} \big(
\frac{\pd F_f^B }{\pd t_k^+} \cdot z^{-k}
+ \frac{k}{2} t_k^+ \cdot z^k
\big),
\end{split}
\end{equation*}
and by M\"obius inversion formula we know that:
\begin{equation*}
G_{f;1,0}^c(z) = G_{f;1,0}(z)
=
\sum_{k>0:\text{ odd}} \big(
\frac{\pd F_f^B }{\pd t_k^+} \cdot z^{-k}
+ \frac{k}{2} t_k^+ \cdot z^k
\big),
\end{equation*}
and the case $(n,m)=(1,0)$ is proved by taking $\bm t=0$.
Similarly,
we have:
\begin{equation*}
G_{f;0,1}^c(z) =
G_{f;0,1}(z)
=
\sum_{k>0:\text{ odd}} \big(
\frac{k}{2} t_k^- \cdot z^{-k}
+ \frac{\pd F_f^B }{\pd t_k^-}  \cdot z^k
\big),
\end{equation*}
which proves the case $(n,m)=(0,1)$.

Now consider the case $(n,m)=(2,0)$.
We have:
\begin{equation*}
\begin{split}
G_{f;2,0} (z_1,z_2)
=& \frac{1}{\tau_f^B}
\sum_{j,k>0:\text{ odd}}
\big(
z_1^{-j}\frac{\pd}{\pd t_j^+} + \frac{j}{2} t_j^+\cdot z_1^j
\big) \big(
z_2^{-k}\frac{\pd}{\pd t_k^+} + \frac{k}{2} t_k^+\cdot z_2^k
\big) \tau_f^B \\
= & \sum_{j,k>0:\text{ odd}} \Big(
\big( \frac{\pd^2 F_f^B}{\pd t_j^+ \pd t_k^+} +
\frac{\pd F_f^B}{\pd t_j^+} \frac{\pd F_f^B}{\pd t_k^+}
\big) z_1^{-j} z_2^{-k}
+ \frac{j}{2}t_j^+ \frac{\pd F_f^B}{\pd t_k^+} z_1^jz_2^{-k}\\
& + \frac{k}{2}t_k^+ \frac{\pd F_f^B}{\pd t_j^+} z_1^{-j}z_2^{k}
+ \delta_{j,k} \cdot \frac{k}{2}z_1^{-j} z_2^k
+ \frac{jk}{4} t_j^+ t_k^+ z_1^jz_2^k
\Big),
\end{split}
\end{equation*}
and then by the M\"obius inversion formula,
\begin{equation*}
\begin{split}
G_{f;2,0}^c (z_1,z_2)
=& G_{f;2,0} (z_1,z_2) - G_{f;1,0} (z_1) G_{f;1,0} (z_2) \\
=& \sum_{j,k>0:\text{ odd}} \frac{\pd^2 F_f^B}{\pd t_j^+ \pd t_k^+}
+ \sum_{j>0: \text{ odd}} \frac{j}{2}z_1^{-j} z_2^j.
\end{split}
\end{equation*}
Then by taking $\bm t=0$ we have proved the case $(2,0)$.
The computations for $(n,m)= (0,2)$ and $(1,1)$ are similar
and here we omit the details.
In the case $(n,m)=(1,1)$ there is no additional term
$\frac{z_1z_2(z_1^2+z_2^2)}{2(z_1^2-z_2^2)^2}$.
\end{proof}

\section{A Formula for the Connected Bosonic $(n,m)$-Point Functions}
\label{sec-main}

In this section,
we derive an explicit formula for the connected bosonic $(n,m)$-point functions
of $\tau_f^B$.
This is  our main result of this paper.

First we need the following combinatorial result
(see \cite[Prop. 4.1]{wy}):
\begin{Proposition}
[\cite{wy}]
Assume $\xi(x,y)$ is a function with $\xi(x,y)=-\xi(y,x)$,
and for each $n\geq 1$ we define an anti-symmetric matrix $M(n)$ of size $2n\times 2n$ by:
\begin{equation*}
M(n)_{ij} =
\xi\big((-1)^i z_{\lceil \frac{i}{2} \rceil }, (-1)^j z_{\lceil \frac{j}{2} \rceil })
\end{equation*}
for $1\leq i<j\leq 2n$.
Define a family of functions $\{\varphi(z_1,\cdots,z_n)\}_{n\geq 1}$ by:
\begin{equation*}
\varphi(z_1,\cdots,z_n) = \Pf(M(n)_{ij})_{1\leq i,j \leq 2n}
\end{equation*}
for every $n$,
then the connected version
\begin{equation*}
\varphi^c (z_1,\cdots,z_n) =
\sum_{I_1\sqcup \cdots \sqcup I_k=[n]}
\frac{(-1)^{k-1}}{k} \varphi(z_{I_1})
\cdots \varphi(z_{I_k}),
\end{equation*}
is given by:
\begin{equation*}
\varphi^c (z_1,\cdots,z_n) = \sum_{\substack{\text{$n$-cycles $\sigma$} \\ \epsilon_2,\cdots,\epsilon_n \in\{\pm 1\}}}
(-\epsilon_2\cdots\epsilon_n) \cdot
\prod_{i=1}^n \xi(\epsilon_{\sigma(i)} z_{\sigma(i)}, -\epsilon_{\sigma(i+1)} z_{\sigma(i+1)}),
\end{equation*}
where we use the conventions
$\epsilon_{1} =1$ and
$\sigma(n+1)=\sigma(1)$.
\end{Proposition}
\begin{Remark}
The above proposition is a Pfaffian-analogue of \cite[Prop. 5.2]{zhou1}.
\end{Remark}

Now we take the anti-symmetric matrix $M(n)$ to be
the matrix $(C_{i,j})$ in Theorem \ref{thm-disconn-b-nmpt}.
Then Theorem \ref{thm-disconn-b-nmpt} and the above proposition
tells us that the connected bosonic $(n,m)$-point function
$\langle H(z_{[n+m]}) \rangle_{f;n,m}^c$
can be represented as a summation over $(n+m)$-cycles.
And by combining this with Proposition \ref{prop-relation-fe},
we obtain:
\begin{Theorem}
\label{thm-main-conn-nmpt}
We have:
\be
\label{eq-mainthm-sumcyclesB}
\begin{split}
&\sum_{\substack{j_1,\cdots,j_n>0:\text{ odd}\\k_1,\cdots,k_m >0:\text{ odd} }}
\frac{\pd^{n+m} F_f^B(\bm t^+,\bm t^-)}
{\pd t_{j_1}^+\cdots\pd t_{j_n}^+ \pd t_{k_1}^-\cdots\pd t_{k_m}^-}
\Big|_{\bm t=0} \cdot
\prod_{a=1}^n z_a^{-j_a} \prod_{b=1}^m z_{n+b}^{k_b}\\
=&
\sum_{\substack{\text{$(n+m)$-cycles $\sigma$} \\ \epsilon_2,\cdots,\epsilon_{n+m} \in\{\pm 1\}}}
(-\epsilon_2\cdots\epsilon_{n+m}) \cdot
\prod_{i=1}^{n+m} \xi(\epsilon_{\sigma(i)} z_{\sigma(i)}, -\epsilon_{\sigma(i+1)} z_{\sigma(i+1)})\\
&- (\delta_{n,2}\delta_{m,0} + \delta_{n,0}\delta_{m,2} )
\cdot i_{z_1,z_2} \frac{z_1z_2 (z_1^2+z_2^2)}{2(z_1^2-z_2^2)^2},
\end{split}
\ee
where for $\sigma(i) < \sigma(i+1)$,
\be
\begin{split}
&\xi(\epsilon_{\sigma(i)} z_{\sigma(i)}, -\epsilon_{\sigma(i+1)} z_{\sigma(i+1)})\\
=& \begin{cases}
A_f \big(
\epsilon_{\sigma(i)} z_{\sigma(i)}, -\epsilon_{\sigma(i+1)} z_{\sigma(i+1)}
\big),
& \text{ if $\sigma(i)\leq n<\sigma(i+1)$;}\\
A_0 \big(
\epsilon_{\sigma(i)} z_{\sigma(i)}, -\epsilon_{\sigma(i+1)} z_{\sigma(i+1)}
\big),
& \text{ otherwise,}
\end{cases}
\end{split}
\ee
and $\xi(\epsilon_{\sigma(i)} z_{\sigma(i)}, -\epsilon_{\sigma(i+1)} z_{\sigma(i+1)})
= - \xi( -\epsilon_{\sigma(i+1)} z_{\sigma(i+1)}, \epsilon_{\sigma(i)} z_{\sigma(i)})$
if $\sigma(i) > \sigma(i+1)$.
Here $A_f$ is the series defined by \eqref{eq-def-Af},
and we use the conventions
$\epsilon_{1} =1$ and
$\sigma(n+m+1)=\sigma(1)$.
\end{Theorem}

The above theorem can be simplified such that one can actually get rid of the choice of
the signs $\epsilon_2,\cdots,\epsilon_{n+m} = \pm 1$.
In fact,
we have the following:
\begin{Theorem}
\label{thm-main-conn-nmpt-2}
We have:
\be
\label{eq-mainthm-sumcyclesB-2}
\begin{split}
&\sum_{\substack{j_1,\cdots,j_n>0:\text{ odd}\\k_1,\cdots,k_m >0:\text{ odd} }}
\frac{\pd^{n+m} F_f^B(\bm t^+,\bm t^-)}
{\pd t_{j_1}^+\cdots\pd t_{j_n}^+ \pd t_{k_1}^-\cdots\pd t_{k_m}^-}
\Big|_{\bm t=0} \cdot
\prod_{a=1}^n z_a^{-j_a} \prod_{b=1}^m z_{n+b}^{k_b}\\
=& -2^{n+m-1} \cdot \Big[
\sum_{\text{$(n+m)$-cycles $\sigma$} }
\prod_{i=1}^{n+m} \xi( z_{\sigma(i)}, - z_{\sigma(i+1)}) \Big]_{\text{odd}}\\
&- (\delta_{n,2}\delta_{m,0} + \delta_{n,0}\delta_{m,2} )
\cdot i_{z_1,z_2} \frac{z_1z_2 (z_1^2+z_2^2)}{2(z_1^2-z_2^2)^2},
\end{split}
\ee
where $[\cdot]_{\text{odd}}$ means taking the terms of odd degrees in every $z_i$;
and for $\sigma(i) < \sigma(i+1)$,
\be
\xi( z_{\sigma(i)}, - z_{\sigma(i+1)})\\
= \begin{cases}
A_f (
 z_{\sigma(i)}, - z_{\sigma(i+1)}),
& \text{ if $\sigma(i)\leq n<\sigma(i+1)$;}\\
A_0 (
 z_{\sigma(i)}, - z_{\sigma(i+1)}),
& \text{ otherwise,}
\end{cases}
\ee
and for $\sigma(i) > \sigma(i+1)$,
\be
\xi( z_{\sigma(i)}, - z_{\sigma(i+1)})\\
= \begin{cases}
-A_f (
  - z_{\sigma(i+1)}, z_{\sigma(i)}),
& \text{ if $\sigma(i)> n\geq \sigma(i+1)$;}\\
-A_0 (
 - z_{\sigma(i+1)},z_{\sigma(i)}),
& \text{ otherwise,}
\end{cases}
\ee
\end{Theorem}
\begin{proof}
For a fixed cycle $\sigma$ and a fixed $j$,
$z_j$ appears only in two terms
\begin{equation*}
\xi (\pm z_i, -\epsilon_j z_j) \cdot \xi (\epsilon_j z_j, \pm z_k)
\end{equation*}
in $\prod_{i=1}^{n+m} \xi(\epsilon_{\sigma(i)} z_{\sigma(i)}, -\epsilon_{\sigma(i+1)} z_{\sigma(i+1)})$
(where $i$ and $k$ are adjacent to $j$ in this cycle $\sigma$),
thus replacing $\epsilon_j$ by $-\epsilon_j$ is equivalent to replacing $z_j$ by $-z_j$.
Then replacing $\epsilon_j$ by $-\epsilon_j$ does not change the terms
with odd orders in $z_j$ in the product
\begin{equation*}
\epsilon_2\cdots\epsilon_{n+m}
\prod_{i=1}^{n+m} \xi(\epsilon_{\sigma(i)} z_{\sigma(i)}, -\epsilon_{\sigma(i+1)} z_{\sigma(i+1)}).
\end{equation*}
Moreover,
we already know that the order of $z_j$ in
the left-hand side of \eqref{eq-mainthm-sumcyclesB} is always an odd number,
thus the conclusion holds by taking the canonical choice of signs
$\epsilon_2 = \cdots = \epsilon_{n+m} =1$ in the right-hand side
and then restricting to terms of odd degrees.
\end{proof}

Similar to the case of the diagonal 2d Toda lattice tau-functions (see \cite[\S 4]{wy2}),
one has the following vanishing property:
\begin{Corollary}
\label{cor-vanish}
One has
\begin{equation*}
\sum_{\substack{j_1,\cdots,j_n>0:\text{ odd}\\k_1,\cdots,k_m >0:\text{ odd} }}
\frac{\pd^{m+n} F_f^B(\bm t^+,\bm t^-)}{\pd t_{j_1}^+ \cdots
\pd t_{j_n}^+ \pd t_{k_1}^- \cdots \pd t_{k_m}^- }
\Big|_{\bm t=0} =0
\end{equation*}
unless $j_1+ j_2 +\cdots +j_n = k_1+k_2+\cdots +k_m$.
\end{Corollary}
\begin{proof}
Recall that $A_0(w,z)$, $A_f(w,z)$, and $i_{w,z} \frac{wz (w^2+z^2)}{2(w^2-z^2)^2}$
are all of the form:
\begin{equation*}
\sum_{k\geq 0} c_k \cdot w^{-k} z^k,
\end{equation*}
thus the total order of non-negative powers equals to the total order of negative powers
in the right-hand side of \eqref{eq-mainthm-sumcyclesB-2}.
\end{proof}

Now we give some examples of the bosonic $(n,m)$-point functions
for small $(n,m)$ using the above theorem.
First notice that by Corollary \ref{cor-vanish} we easily see:
\begin{Corollary}
\label{cor-vanish-2}
For $n,m\geq 1$,
we have:
\be
\begin{split}
&\frac{\pd^{n} F_f^B(\bm t^+,\bm t^-)}
{\pd t_{j_1}^+\cdots\pd t_{j_n}^+ }=0,
\qquad \forall j_1,\cdots,j_n>0:\text{ odd};\\
&\frac{\pd^{m} F_f^B(\bm t^+,\bm t^-)}
{ \pd t_{k_1}^-\cdots\pd t_{k_m}^-} =0,
\qquad \forall k_1,\cdots,k_m >0:\text{ odd}.
\end{split}
\ee
\end{Corollary}

\begin{Example}
By Corollary \ref{cor-vanish-2} we know that the first non-trivial example is $(n,m) = (1,1)$.
In this case,
there is only one $2$-cycle $\sigma = (12)$,
thus by Theorem \ref{thm-main-conn-nmpt} or \ref{thm-main-conn-nmpt-2} we have:
\begin{equation*}
\begin{split}
& \sum_{j,k>0:\text{ odd} }
\frac{\pd^{2} F_f^B(\bm t^+,\bm t^-)}
{\pd t_j^+ \pd t_{k}^-}
\Big|_{\bm t=0} \cdot
 z_1^{-j} z_2^{k} \\
 =&
A_f(z_1,-z_2) A_f(-z_1,z_2) - A_f(z_1,z_2) A_f(-z_1,-z_2)\\
=& 2 \big[ A_f(z_1,-z_2) A_f(-z_1,z_2) \big]_{\text{odd}}.
\end{split}
\end{equation*}
\end{Example}

\begin{Example}
Now consider the case $(n,m) = (2,1)$.
There are two $3$-cycle $\sigma = (123)$, $(132)$.
By Theorem \ref{thm-main-conn-nmpt-2} we have:
\begin{equation*}
\begin{split}
& \sum_{j,k,l>0:\text{ odd} }
\frac{\pd^{3} F_f^B(\bm t^+,\bm t^-)}
{\pd t_j^+  \pd t_{k}^+ \pd t_l^- }
\Big|_{\bm t=0} \cdot
 z_1^{-j} z_2^{-k} z_3^l \\
 =& \big[
4A_0(z_1,-z_2) A_f(z_2,-z_3) A_f(-z_1,z_3)
 - 4A_f(z_1,-z_3) A_f(-z_2,z_3) A_0(-z_1,z_2)
 \big]_{\text{odd}}.
\end{split}
\end{equation*}
In this case the above summation vanishes identically due to Corollary \ref{cor-vanish}.
\end{Example}

\begin{Example}
For $n+m=4$,
there are six $4$-cycles $\sigma= (1234)$, $(1243)$, $(1324)$, $(1342)$, $(1423)$, $(1432)$.
Then for $(n,m)=(3,1)$ we have:
\begin{equation*}
\begin{split}
& \sum_{i,j,k,l>0:\text{ odd} }
\frac{\pd^{4} F_f^B(\bm t^+,\bm t^-)}
{\pd t_i^+ \pd t_j^+  \pd t_{k}^+ \pd t_l^- }
\Big|_{\bm t=0} \cdot
z_1^{-i} z_2^{-j} z_3^{-k} z_4^l \\
 =
& \big[ 8 A_0(z_1,-z_2) A_0(z_2,-z_3) A_f(z_3,-z_4) A_f(-z_1,z_4) \\
& -8 A_0(z_1,-z_2) A_f(z_2,-z_4) A_f(-z_3,z_4) A_0(-z_1,z_3) \\
& -8 A_0(z_1,-z_3) A_0(-z_2,z_3) A_f(z_2,-z_4) A_f(-z_1,z_4) \\
& -8 A_0(z_1,-z_3) A_f(z_3,-z_4) A_f(-z_2,z_4) A_0(-z_1,z_2) \\
& -8 A_f(z_1,-z_4) A_f(-z_2,z_4) A_0(z_2,-z_3) A_0(-z_1,z_3) \\
& +8 A_f(z_1,-z_4) A_f(-z_3,z_4) A_0(-z_2,z_3) A_0(-z_1,z_2)\big]_{\text{odd}}. \\
\end{split}
\end{equation*}
And for $(n,m)=(2,2)$ we have:
\begin{equation*}
\begin{split}
& \sum_{i,j,k,l>0:\text{ odd} }
\frac{\pd^{4} F_f^B(\bm t^+,\bm t^-)}
{\pd t_i^+ \pd t_j^+  \pd t_{k}^- \pd t_l^- }
\Big|_{\bm t=0} \cdot
z_1^{-i} z_2^{-j} z_3^{k} z_4^l \\
 =
& \big[ 8 A_0(z_1,-z_2) A_f(z_2,-z_3) A_0(z_3,-z_4) A_f(-z_1,z_4) \\
& -8 A_0(z_1,-z_2) A_f(z_2,-z_4) A_0(z_3,-z_4) A_f(-z_1,z_3) \\
& -8 A_f(z_1,-z_3) A_f(-z_2,z_3) A_f(z_2,-z_4) A_f(-z_1,z_4) \\
& -8 A_f(z_1,-z_3) A_0(z_3,-z_4) A_f(-z_2,z_4) A_0(-z_1,z_2) \\
& -8 A_f(z_1,-z_4) A_f(-z_2,z_4) A_f(z_2,-z_3) A_f(-z_1,z_3) \\
& +8 A_f(z_1,-z_4) A_0(-z_3,z_4) A_f(-z_2,z_3) A_0(-z_1,z_2)
\big]_{\text{odd}}. \\
\end{split}
\end{equation*}
\end{Example}

\section{Application to Connected Spin Double Hurwitz Numbers}
\label{sec-spinH}

In this section,
we apply Theorem \ref{thm-main-conn-nmpt} to the computations of
connected spin double Hurwitz numbers.

The ordinary Hurwitz numbers \cite{hur} count the numbers of branched covers between
Riemann surfaces with specified ramification types.
In \cite{eop},
Eskin-Okounkov-Pandharipande introduced a new type of Hurwitz numbers
called spin Hurwitz numbers,
by introducing a spin structure (or theta characteristic) on the source.
Similar to the case of ordinary Hurwitz numbers (see \cite{pa, Ok1}),
the generating series of disconnected spin Hurwitz numbers are controlled by some integrable hierarchies,
see \cite{le, gkl}.
In particular,
Giacchetto-Kramer-Lewa\'nski \cite{gkl} showed that the generating series
\be
\tau^{r,\vartheta}(\bm t^+,\bm t^-) = \sum_{g;\mu^+,\mu^-}
2^{g-1} \beta^b h_{g;\mu^+,\mu^-}^{\bullet, r,\vartheta}
\frac{p^+_{\mu^+} p^-_{\mu^-}}{l(\mu^+)! l(\mu^-)!}
\ee
of disconnected spin double Hurwitz numbers $h_{g;\mu^+,\mu^-}^{\bullet, r,\vartheta} $
with $(r+1)$-completed cycles (where $r$ is even)
is a tau-function of the 2-BKP hierarchy,
where the summation is over all integers $g$ (genus) and odd partitions $\mu^\pm$.
Here for a partition $\mu = (\mu_1,\cdots,\mu_l)$,
the number $l(\mu) = l$ denotes its length,
and
\begin{equation*}
p_\mu^\pm = p_{\mu_1}^\pm p_{\mu_2}^\pm \cdots p_{\mu_l}^\pm,
\qquad
\text{ where $p_n^\pm = n\cdot t_n^\pm$.}
\end{equation*}
The number $b$ of $(r+1)$-completed cycle ramification points
is determined by the Riemann-Hurwitz formula:
\be
\label{eq-R-H}
b= (2g-2 +l(\mu^+) +l(\mu^-)) /r.
\ee
They found a fermonic representation of this tau-function
(see \cite[Theorem 6.20]{gkl};
see also \cite{le}):
\be
\tau^{r,\vartheta} =
\langle \Gamma_+ (\bm t^+)
 \exp\big(\beta\frac{\hF_{r+1}}{r+1}\big)
 \Gamma_-(\bm t^-) \rangle,
\ee
where $\hF_{r+1}$ is the following operator on the fermionic Fock space:
\be
\hF_{r+1}
= \sum_{k>0} (-1)^k k^{r+1} :\phi_k\phi_{-k}:.
\ee

Now our goal is to compute the connected spin double Hurwitz numbers
with $(r+1)$-completed cycles
$h_{g;\mu^+,\mu^-}^{\circ, r,\vartheta}$,
whose generating series
\be
\sum_{g;\mu^+,\mu^-}
2^{g-1} \beta^b h_{g;\mu^+,\mu^-}^{\circ, r,\vartheta}
\frac{p^+_{\mu^+} p^-_{\mu^-}}{l(\mu^+)! l(\mu^-)!}
\ee
is the free energy $\log \tau^{r,\vartheta}$.
In this case,
we take the function $f$ (see \S \ref{sec-ferm-2n2m}) to be:
\be
f^{r,\vartheta}(k) = \beta\cdot \frac{k^{r+1}}{r+1},
\ee
then $\hf^{r,\vartheta} = \beta \frac{\hF_{r+1}}{r+1}$.
In this case,
the series $A_f(z,w)$ defined by \eqref{eq-def-Af} is:
\be
\label{eq-def-Af-rtheta}
A_{f^{r,\vartheta}} (w,z) = -\frac{1}{4} -\half \sum_{k=1}^\infty
(-1)^k e^{ \beta \frac{k^{r+1}}{r+1} }w^{-k}z^k.
\ee
Denote:
\be
h_{\mu^+,\mu^-}^{\circ, r,\vartheta}(\beta) =
\sum_b 2^{g-1} \beta^b
\frac{h_{g;\mu^+,\mu^-}^{\circ, r,\vartheta}}
{l(\mu^+)! l(\mu^-)!},
\ee
then by Theorem \ref{thm-main-conn-nmpt} we have:
\begin{Theorem}
\label{thm-spinH}
Let $\mu^+ =(\mu_1^+,\cdots,\mu_n^+)$ and $\mu^- = (\mu_1^-,\cdots,\mu_m^-)$
be two odd partitions,
then we have:
\begin{equation*}
h_{\mu^+,\mu^-}^{\circ, r,\vartheta}(\beta) = -
\frac{2^{n+m-1}}{Z_{\mu^+} Z_{\mu^-}}
\text{Coeff}_{\prod\limits_{a=1}^n z_a^{-\mu_a^+} \prod\limits_{b=1}^m z_{n+b}^{\mu_b^-}}
\bigg(
\sum_{\text{$(n+m)$-cycles} }
\prod_{i=1}^{n+m} \xi( z_{\sigma(i)}, - z_{\sigma(i+1)})
\bigg),
\end{equation*}
where $\text{Coeff}$ means taking the coefficient,
and for a partition $\mu=(1^{m_1(\mu)} 2^{m_2(\mu)}\cdots)$ we denote:
\be
Z_\mu = \prod_{j\geq 1} m_j(\mu)! \cdot j^{m_j(\mu)}.
\ee
Here $\xi(\epsilon_{\sigma(i)} z_{\sigma(i)}, -\epsilon_{\sigma(i+1)} z_{\sigma(i+1)})$ are as follows:
for $\sigma(i) < \sigma(i+1)$,
\be
\xi( z_{\sigma(i)}, - z_{\sigma(i+1)})\\
= \begin{cases}
A_{f^{r,\vartheta}} \big(
 z_{\sigma(i)}, - z_{\sigma(i+1)}
\big),
& \text{ if $\sigma(i)\leq n<\sigma(i+1)$;}\\
A_0 \big(
 z_{\sigma(i)}, - z_{\sigma(i+1)}
\big),
& \text{ otherwise.}
\end{cases}
\ee
And $\xi( z_{\sigma(i)}, - z_{\sigma(i+1)})
= - \xi( - z_{\sigma(i+1)},  z_{\sigma(i)})$
if $\sigma(i) > \sigma(i+1)$.
Here $A_{f^{r,\vartheta}}$ is given by \eqref{eq-def-Af-rtheta},
and we use the convention $\sigma(n+m+1)=\sigma(1)$.
\end{Theorem}

\begin{Example}
For $(n,m)=(1,1)$,
we have:
\begin{equation*}
\begin{split}
 h_{(u),(v)}^{\circ,r,\vartheta}(\beta)
= & \frac{1}{uv} \text{Coeff}_{z_1^{-u}z_2^v}
\big(
2A_{f^{r,\vartheta}}(z_1,-z_2) A_{f^{r,\vartheta}}(-z_1,z_2)
\big)\\
= & \frac{\delta_{u,v}}{4uv} \big( 1-(-1)^u \big) \Big(
e^{\beta \frac{u^{r+1}}{r+1}}
+ \sum_{j=1}^{u-1} e^{\beta( \frac{j^{r+1}+(u-j)^{r+1}}{r+1} ) }
\Big)\\
=& \frac{\delta_{u,v}}{2uv} \Big(
e^{\beta \frac{u^{r+1}}{r+1}}
+ \sum_{j=1}^{u-1} e^{\beta( \frac{j^{r+1}+(u-j)^{r+1}}{r+1} ) }
\Big),
\end{split}
\end{equation*}
since $u$ is odd.
Or more explicitly,
\begin{equation*}
h_{g;(u),(u)}^{\circ,r,\vartheta} =
 \frac{1}{2^g \cdot u^2 \cdot b!(r+1)^b} \Big( u^{b(r+1)}
 +\sum_{j=1}^u \big( j^{r+1} + (u-j)^{r+1} \big)^b \Big),
\end{equation*}
where $b$ and $g$ are related by \eqref{eq-R-H}.
\end{Example}

\begin{Example}
Consider $(n,m) = (2,2)$.
Let $\mu^+ = (u_1,u_2)$, $\mu^- =(v_1,v_2)$ be two odd partitions with
$|\mu^+| = |\mu^-|$,
then Theorem \ref{thm-spinH} gives
(here we omit the details,
and for simplicity we have assumed $u_1\leq v_1$ in the computations):
\begin{equation*}
\begin{split}
& h_{(u_1,u_2),(v_1,v_2)}^{\circ,r,\vartheta}(\beta)
=  \frac{1}{(1+\delta_{u_1,u_2})(1+\delta_{v_1,v_2})u_1 u_2 v_1 v_2} \times\\
&\quad
\Big( \sum_{i=1}^{v_2-1} e^{\frac{\beta}{r+1} ( i^{r+1} + (u_1+u_2-i)^{r+1} ) }
+  (\half - \frac{\delta_{u_1,v_1}}{4}) e^{\frac{\beta}{r+1} ( v_1^{r+1}+v_2^{r+1} ) }
+  \half e^{\frac{\beta}{r+1} (u_1+u_2)^{r+1} }
 \\
&\quad
 - \half \sum_{l=1}^{v_2-1} e^{\frac{\beta}{r+1} ( l^{r+1} +(u_1-l)^{r+1}+ (v_2-l)^{r+1}+ (v_1-u_1+l)^{r+1} ) }\\
&\quad
 - \half \sum_{j=1}^{v_2-1} e^{\frac{\beta}{r+1} ( j^{r+1} +(u_2-j)^{r+1}+ (v_2-j)^{r+1}+ (v_1-u_2+j)^{r+1} ) }\\
&\quad
 -\frac{1}{2}e^{\frac{\beta}{r+1} ((u_2-v_2)^{r+1} + u_1^{r+1} + v_2^{r+1})}
-\frac{1}{2}e^{\frac{\beta}{r+1} ((v_1-u_2)^{r+1} + u_2^{r+1} + v_2^{r+1})}
\Big).
\end{split}
\end{equation*}

\end{Example}

\vspace{.2in}
{\em Acknowledgements}.
We thank the anonymous referee for helpful suggestions that improve the presentation of this paper.
We also thank Professor Huijun Fan, Professor Xiaobo Liu, and Professor Jian Zhou for encouragement.


\begin{thebibliography}{9}
%\addcontentsline{toc}{chapter}{Bibliography}


\bibitem{AS}Alexandrov A, Shadrin S. Elements of spin Hurwitz theory: closed algebraic formulas, blobbed topological recursion, and a proof of the Giacchetto-Kramer-Lewanski conjecture. arXiv preprint arXiv:2105.12493, 2021.

\bibitem{by}Balogh F, Yang D. Geometric interpretation of Zhou's explicit formula for the Witten-Kontsevich tau function. Letters in Mathematical Physics, 2017, 107(10):1837-1857.

\bibitem{djkm}Date E, Jimbo M, Kashiwara M, Miwa T. Transformation groups for soliton equations IV. A new hierarchy of soliton equations of KP-type. Physica D, 1982, 4(3):343-365.

\bibitem{djm}Date E, Jimbo M, Miwa T. Solitons: Differential equations, symmetries and infnite dimensional algebras. Cambridge University Press, 2000.

\bibitem{dij}Dijkgraaf R. Mirror Symmetry and Elliptic Curves. Moduli Space of Curves, 1995, 129:149-164.

\bibitem{elsv1}Ekedahl T, Lando S, Shapiro M, Vainshtein A. On Hurwitz numbers and Hodge integrals. 1999, 328(12):1175-1180.

\bibitem{elsv2}Ekedahl T, Lando S, Shapiro M, Vainshtein A. Hurwitz numbers and intersections on moduli spaces of curves. Inventiones Mathematicae, 2001, 146(2): 297-327.

\bibitem{eop}Eskin A, Okounkov A, Pandharipande R. The theta characteristic of a branched covering. Adv. Math., 2008, 217: 873-888.

\bibitem{gkl}Giacchetto A, Kramer R, Lewa\'nski D. A new spin on Hurwitz theory and ELSV via theta characteristics. arXiv preprint arXiv:2104.05697, 2021.

\bibitem{gkls}Giacchetto A, Kramer R, Lewa\'nski D, Sauvaget A. The Spin Gromov-Witten/Hurwitz correspondence for $\bP^1$. arXiv preprint arXiv:2208.03259, 2022.

\bibitem{gjv}Goulden I P, Jackson D M, Vakil R. Towards the geometry of double Hurwitz numbers. Adv. Math., 2005, 198(1): 43-92.

\bibitem{gv}Graber T, Vakil R. Hodge integrals and hurwitz numbers via virtual localization. Compositio Math, 2003, 135(1):25-36.

\bibitem{gu}Gunningham S. Spin Hurwitz numbers and topological quantum field theory. Geometry \& Topology, 2016, 20: 1859-1907.

\bibitem{hb}Harnad J, Balogh F. Tau Functions and their Applications. Cambridge Monographs on Mathematical Physics. Cambridge University Press, 2021.

\bibitem{hur}Hurwitz A. \"Uber die Anzahl der Riemann'schen Fl\"achen mit gegebenen Verzweigungspunkten. Math. Ann., 1902, 55:53-66.

\bibitem{jm}Jimbo M, Miwa T. Solitons and Infinite-Dimensional Lie Algebras. Publications of the Research Institute for Mathematical Sciences, 1983, 1983(19): 943-1001.

\bibitem{jo}Johnson P. Double Hurwitz numbers via the infinite wedge. Transactions of the American Mathematical Society, 2015, 367(9): 6415-6440.

\bibitem{le2}Lee J. A note on Gunningham's formula. Bull. Aust. Math. Soc, 2018, 98:389-401.

\bibitem{le}Lee J. A square root of Hurwitz numbers. Manuscripta Mathematica, 2020,  162(1-2): 99-113.

\bibitem{mmn}Mironov A, Morozov A, Natanzon S. Cut-and-join structure and integrability for spin Hurwitz numbers. European Physical Journal C, 2020, 80(2).

\bibitem{mmno}Mironov A, Morozov A, Natanzon S, Orlov A. Around spin Hurwitz numbers. Letters in Mathematical Physics, 2021, 111:124.

\bibitem{Ok1}Okounkov A. Toda equations for Hurwitz numbers. Math. Res. Lett., 2000, 7:447-453.

\bibitem{op}Okounkov A, Pandharipande R. Gromov-Witten theory, Hurwitz theory, and completed cycles. Annals of Mathematics, 2006, 163(2):517-560.

\bibitem{op2}Okounkov A, Pandharipande R. The equivariant Gromow-Witten theory of $\bP^1$. Annals of Mathematics, 2006, 163(2):561-605.

\bibitem{ost}Orlov A Y, Shiota T, Takasaki K. Pfaffian structures and certain solutions to BKP hierarchies I. Sums over partitions. arXiv preprint arXiv:1201.4518, 2012.

\bibitem{pa}Pandharipande R. The Toda equations and the Gromov-Witten theory of the Riemann sphere. Letters in Mathematical Physics, 2000, 53(1):59-74.

\bibitem{sa}Sato M. Soliton Equations as Dynamical Systems on an Infinite Dimensional Grassmann Manifold. RIMS Kokyuroku, 1981, 439:30-46.

\bibitem{sw}Segal G, Wilson G. Loop Groups and Equations of KdV Type. Publications Math\'ematiques de l'IH\'ES, 1985, 61(1):5-65.

\bibitem{ssz}Shadrin S, Spitz L, Zvonkine D. On double Hurwitz numbers with completed cycles. Journal of the London Mathematical Society, 2012, 86(2):407-432.

\bibitem{ta}Takebe T. Toda lattice hierarchy and conservation laws. Communications in Mathematical Physics, 1990, 129(2):281-318.

\bibitem{tu}Tu M H. On the BKP hierarchy: additional symmetries, Fay identity and Adler-Shiota-van Moerbeke formula. Letters in Mathematical Physics, 2007, 81(2):93-105.

\bibitem{ut}Ueno K, Takasaki K. Toda lattice hierarchy. Advanced Studies in Pure Mathematics 4, 1984:1-95.

\bibitem{va}van de Leur J. The Adler-Shiota-van Moerbeke formula for the BKP hierarchy. Journal of Mathematical Physics, 1995, 36:4940-4951.

\bibitem{wy}Wang Z, Yang C. BKP Hierarchy, Affine Coordinates, and a Formula for Connected Bosonic $N$-Point Functions. Letters in Mathematical Physics, 2022, 112:62.

\bibitem{wy2}Wang Z, Yang C. Diagonal Tau-Functions of 2D Toda Lattice Hierarchy, Connected $(n,m)$-Point Functions, and Double Hurwitz Numbers. arXiv preprint arXiv:2210.08712, 2022.

\bibitem{yo}You Y. Polynomial solutions of the BKP hierarchy and projective representations of symmetric groups. Infinite-dimensional Lie algebras and groups, Luminy-Marseille, Adv. Ser. Math. Phys., Volume 7: 449-464.

\bibitem{zhou3}Zhou J. Explicit Formula for Witten-Kontsevich Tau-Function. arXiv preprint arXiv:1306.5429, 2013.

\bibitem{zhou1}Zhou J. Emergent geometry and mirror symmetry of a point. arXiv preprint arXiv:1507.01679, 2015.





\end{thebibliography}
\end{document}